\documentclass[a4paper,11pt]{article}
\usepackage[T1]{fontenc}
\usepackage[utf8]{inputenc}
\usepackage{lmodern}
\usepackage{amsfonts}
\usepackage{amssymb}
\usepackage{amsmath}
\usepackage{fullpage}
\usepackage{amsthm}
\usepackage[us,hhmmss]{datetime}
\usepackage{cite}
\usepackage{authblk}
\usepackage{graphicx}

\numberwithin{equation}{section}	

\title{Microscopic derivation of the Fr\"ohlich Hamiltonian for the Bose polaron in the mean-field limit}
\author{Krzysztof My\'{s}liwy\footnote{krzysztof.mysliwy@ist.ac.at}\ \ and Robert Seiringer\footnote{robert.seiringer@ist.ac.at}} 
\affil{\textit{IST Austria, Am Campus 1, 3400 Klosterneuburg, Austria}}

\newtheorem{thm}{Theorem}
\newtheorem{proposition}[thm]{Proposition}
\newtheorem{lemma}[thm]{Lemma}

\newtheorem{condition}{Assumption}

\theoremstyle{definition}
\newtheorem{rem}{Remark}[thm]
\begin{document}

\maketitle

\begin{abstract}
We consider the quantum mechanical many-body problem of a single impurity particle immersed in a weakly interacting Bose gas. The impurity  interacts with the bosons via a two-body potential. We study the Hamiltonian of this system in the mean-field limit and rigorously show that, at low energies, the problem is well described by the Fr\"ohlich polaron model.
\end{abstract}

\section{Introduction and main results}

\subsection{The polaron}
The behavior of impurity particles interacting with a large background constitutes an important class of problems within condensed matter physics \cite{Dev, Weiss}. Among these, one of the most prominent is the polaron problem, where one considers a quantum particle of mass $M$ linearly coupled to a scalar boson field. 
For a translation invariant system, this corresponds to the formal Hamiltonian \begin{equation}\label{Fro1}
  H=\frac{P^2}{2M}+\sum_k e_k a^{\dagger}_k a_k +\sum_k \left(g_k a_k e^{ikR}+g^*_k a^{\dagger}_{k}e^{-ikR}\right),
\end{equation} 
where $R$ denotes the position of the impurity particle, and $k$ labels the momentum modes of the field. Moreover, $P=-\mathrm{i}\nabla_R$ is the particle's momentum operator in the canonical representation, and $a^{\dagger}_k,a_k$ are the usual field mode creation and annihilation operators. They satisfy the canonical commutation relations $[a_k,a^{\dagger}_{k'}]=\delta_{k,k'}, [a_k, a_{k'}]=0$. The $g_k$ are coefficients quantifying the coupling of the particle  to the field, with $^*$ denoting the complex conjugate, and $e_k$ is the free field dispersion relation. The natural domain of this Hamiltonian lies in the Hilbert space $\mathcal{H}\otimes \mathcal{F}(\mathcal{K})$, where $\mathcal{H}$ is the Hilbert space of the particle and $\mathcal{K}$ is the Hilbert space of a single field mode, with $\mathcal{F}(\mathcal{K})$ denoting the symmetric Fock space  over $\mathcal{K}$. $\mathcal{K}$ and $\mathcal{H}$ are appropriate $L^2$ spaces, whose exact specification depends on the underlying physical situation; our choice thereof is discussed below. 

The Hamiltonian \eqref{Fro1} is commonly referred to as the Fr\"ohlich Hamiltonian, as it was introduced by Fr\"ohlich in 1937 \cite{Fro} in order to describe electronic motion in polar crystals. 
The \emph{polaron} in this context refers to the picture of an electron dressed with the emerging optical phonons dragged along as it moves. 
 Later, this concept was extended to include other phenomena related to mobile impurities coupled to excitations of the background, giving rise to interesting effects in many materials \cite{Dev,RevPolCrys,Misza} which are still the subject of ongoing research \cite{Nagoasa, Koepsell}. 

In this work, we are interested in a rigorous justification of the use of Hamiltonians of the type \eqref{Fro1} as an effective description of a full quantum mechanical many-body problem. In the case of the original Fr\"ohlich  
model this task seems too ambitious due to a complicated microscopic structure of the background (see, however, \cite{LewinRougerie}, where the classical approximation to the original  
polaron problem, the Pekar functional, is rigorously derived from a specific model of an electron moving through a quantum crystal). The applicability of the polaron picture is not limited to electrons in crystal lattices, however. In fact, recent progress in experiments with ultracold atoms opened the possibility of studying impurity atoms immersed in an environment consisting of many bosonic atoms 
at low temperatures, displaying  Bose--Einstein condensation. As discussed below, at sufficiently low energies the excitations of the bosonic bath correspond to quantized acoustic phonons, and hence the \emph{Bose polaron} corresponds to the impurity atom dressed with these phonons. We refer to \cite{Grusdt} for a review of recent theoretical progress concerning the application of Fr\"ohlich Hamiltonians to these systems. As the mathematical description of cold Bose systems, and  in particular the structure of their excitation spectra at low energies, have  recently been studied rigorously in numerous works \cite{LSSY, LiebYngvason, LNSS, Se11, Marcin, GreSe}, we find it natural to provide a rigorous microscopic derivation of \eqref{Fro1} based on these  results.

\subsection{The $N+1$ Bose gas}
We consider a system of $N$ bosons of mass 1/2 and one additional particle (of an unspecified type of statistics) of mass $M$, all confined to move on the unit torus in $d$ dimensions, $\mathbb{T}^d$. 
\begin{condition}[Assumptions on the potentials]\label{assum}  We assume that 
\begin{enumerate}
\item the bosons interact among themselves via a two-body potential $v: \mathbb{T}^d\rightarrow \mathbb{R}$ which is bounded, Borel measurable, even and of positive type, i.e., all its Fourier coefficients  $v_p$ are non-negative. 
\item the additional impurity particle interacts with the bosons via a real-valued two-body potential $w: \mathbb{T}^d\rightarrow \mathbb{R}$, which is bounded, Borel measurable and even. 
\end{enumerate}
\end{condition}
Note that no assumption is made on the Fourier coefficients $w_p$ of $w$. Nevertheless $w$ being even implies $w_p=w_{-p}\in \mathbb{R}$. Without loss of generality, we may in addition assume that  $v$ and $w$ are non-negative, since they can be shifted by a constant otherwise.

The positions of the bosons are labeled by $\lbrace x_i \rbrace_{i=1}^N, x_i\in \mathbb{T}^d$ and the position of the impurity by $R\in \mathbb{T}^d$. The Hamiltonian of this system reads 
\begin{equation}\label{ham1}
\frac{-\triangle_R}{2M}-\sum_{i=1}^N \triangle_{x_i}+\lambda\sum_{1\leq i<j\leq N}v(\eta (x_i-x_j))+\mu \sum_{i=1}^N w(\nu(x_i-R))
\end{equation} 
where we introduced some coupling ($\lambda$,$\mu$) and scaling ($\eta$,$\nu$) parameters to be chosen. 
It acts on $L^2(\mathbb{T}^d)\otimes \mathcal{H}_N$ with $\mathcal{H}_N$ being the Hilbert space of square-integrable symmetric functions on $\mathbb{T}^{dN}$. Here, $\triangle_y$ denotes the $d-$dimensional Laplacian in the coordinate $y$ acting on functions on the unit torus. The coupling parameters $\lambda$ and $\mu$ determine the  strength of the potentials $v$ and $w$ (for the functional forms of $v$ and $w$ being fixed), whereas $\eta$ and $\nu$ determine the respective ranges (relative to the system size). They can be adjusted to consider various scaling regimes.  The usual thermodynamic limit corresponds to the choice $\eta \sim \nu \sim N^{1/d}$ and $\lambda\sim \mu \sim N^{2/d}$. 
In contrast, we consider here the mean-field limit, where the interactions are weak and extend over the entire system. In particular, we choose  $\lambda=(N-1)^{-1}$, $\mu=N^{-1/2}$, and $\eta=\nu=1$.  For systems without  impurity, this was the  scaling for which the first rigorous results on the excitation spectrum were obtained \cite{Se11,LNSS,GreSe,NamSe}, and our analysis is based on them. The choice $\mu = N^{-1/2}$ for the impurity-boson coupling turns out to be a natural in the analysis, compatible with the methods from \cite{Se11, LNSS} we use, as explained below (see, in particular, Remark~\ref{rem1}). Therefore, from now on we consider the Hamiltonian 
\begin{equation}\label{ham}
  H_N:=\frac{-\triangle_R}{2M}-\sum_{i=1}^N \triangle_{x_i}+\frac{1}{N-1}\sum_{1\leq i<j\leq N}v(x_i-x_j)+
\frac{1}{\sqrt{N}} \sum_{i=1}^N w(x_i-R)
\end{equation}
on $L^2(\mathbb{T}^d)\otimes \mathcal{H}_N$, with $v$ and $w$ non-negative $1$-periodic functions satisfying Assumption~\ref{assum}.

\subsubsection{Motivation of the Fr\"ohlich Hamiltonian}

With $v_p$ and $w_p$ denoting the Fourier coefficients of $v$ and $w$, respectively, 
the second-quantized version of $H_N$ in \eqref{ham} reads 
\begin{equation}\label{sec}
\frac{-\triangle_R}{2M}+E_{\rm H}(N)+\sum_{p\neq 0} p^2 a_p^{\dagger}a_p + \frac{1}{2(N-1)}  \sum_{\substack{p,q,k\in (2\pi \mathbb{Z})^d \\ p\neq 0}} a^{\dagger}_{p+k}a^{\dagger}_{q-p}a_q a_k +\frac{1}{\sqrt{N}} \sum_{\substack{p, k \in(2\pi \mathbb{Z})^d \\ p \neq 0}} w_{p} e^{-ipR} a^{\dagger}_{p+k}a_{k}.
\end{equation} 
We defined the \emph{Hartree ground state energy} \begin{equation}\label{EH}
E_{\rm H}(N)= \frac{N}{2}v_0+\sqrt{N}w_0,
\end{equation} which captures the effect of interactions between particles in the $p=0$ mode. The sums run over $(2\pi\mathbb{Z})^d$ with $p= 0$ excluded. 
Here, $a_p$ denotes the usual annihilation operator $\mathcal{H}_N\rightarrow  \mathcal{H}_{N-1}$ acting as 
\begin{equation}\label{animp}
(a_p \Psi)(x_1,x_2,\cdots,x_{N-1})=\sqrt{N}\int_{\mathbb{T}^{d}}\Psi(x_1,\cdots, x_{N-1},x)e^{-ipx}dx.
\end{equation} 
The second-quantized Hamiltonian \eqref{sec} acts on $L^2(\mathbb{T}^d)\otimes \mathcal{F}$, with $\mathcal{F}$ the bosonic Fock space $\mathcal{F}$  over $L^2(\mathbb{T}^d)$, i.e.,  $\mathcal{F}:=\bigoplus_{i=0}^{\infty}\mathcal{H}_i$ (with $\mathcal{H}_0=\mathbb{C}$). Actually, it preserves $L^2(\mathbb{T}^d) \otimes \mathcal{H}_N$.  
For the system without  impurity, it was predicted by Bogoliubov \cite{Bog47} that for sufficiently low energies, the excitation spectrum of $H_N$ should be composed of elementary excitations, which are physically interpreted as quantized (acoustic) free phonons.  This serves as the basis for the microscopic explanation of the onset of superfluid behavior in low-temperature bosonic systems. From the formal perspective, it provides a specific example of the emergence of an effective quantum field theoretical description of a many-body system. The low-energy effective theory  is predicted to be that of the Hamiltonian 
\begin{equation}\label{bog}
\mathbb{{H}}^{\text{B}}=\sum_{p\neq 0} e_p b^{\dagger}_p b_p\,.
\end{equation} 
Here, $b^{\dagger}_p=\alpha_p a^{\dagger}_p +\beta_p a_{-p}$ where $\alpha_p, \beta_p$ are appropriate constants chosen such that $[b_p,b^{\dagger}_q]=\delta_{p,q}$. Explicitly, $\alpha_p=(1-\gamma_p)^{-1/2}$ with $\gamma_p=1+\frac{p^2-e_p}{v_p}$ and $\beta_p=\gamma_p\alpha_p$. These algebraic relations are realized via a suitable unitary (Bogoliubov) transformation. From \eqref{bog} we deduce that, for low energies, the excitation spectrum is expected to be composed of free bosonic quasi-particles with dispersion relation $e_p$. In the mean-field scaling $\lambda=(N-1)^{-1}$ considered here, one can prove \cite{Se11} that $e_p=\sqrt{p^4+2v_p p^2}$. 
Additionally, it can be shown that in this scaling the ground state energy equals $\frac{1}{2}Nv_0+E^{\rm B}+o(1)$ with the constant $E^{\rm B}$ equal to 
\begin{equation}
E^{\rm B}=-\frac{1}{2}\sum_{p\neq 0} \left( p^2+v_p-\sqrt{p^4+2p^2v_p}\right). 
\end{equation} 

The method employed by Bogoliubov  leading to $\mathbb{{H}}^\text{B}$ consists of the following steps:
\begin{enumerate}
  \item the operators $a_0, a_0^{\dagger}$ are replaced by the number $\sqrt{N}$
  \item all the terms of higher order than quadratic in  creation and annihilation operators that remain in the Hamiltonian are dropped. 
  \end{enumerate}
This procedure is physically motivated by the expectation that for sufficiently small energies there is Bose--Einstein condensation in the system, that is,  the $p=0$ mode is occupied by an overwhelming fraction of particles. Whereas this has not been proven for a generic bosonic system with general interactions, the validity of the Bogoliubov approximation has been rigorously verified (in the case $w\equiv 0$) for a variety of assumptions on $v$ \cite{Se11, Marcin, LNSS, GreSe,BoccatoBrenneckeSchleinCenatiempo}. The first such result \cite{Se11} refers precisely to our conditions on $v$ and, as already mentioned, the mean-field scaling $\lambda=(N-1)^{-1}$, which corresponds to a very weak and long-ranged potential. 

If one applies the Bogoliubov approximation to the Hamiltonian \eqref{sec} with impurity, one expects that the system is, for small energies, effectively described by  the Fr\"ohlich Hamiltonian 
\begin{equation}\label{Fro3}
\mathbb{H}^{\text{F}}:=\frac{-\triangle_R}{2M}+\sum_{p\neq 0} (p^2+v_p) a^{\dagger}_p a_p +\frac{1}{2}\sum_{p\neq 0} v_p (a^{\dagger}_p a^{\dagger}_{-p}+a_pa_{-p})+\sum_{p\neq 0} w_p e^{-ipR}(a^{\dagger}_p+a_{-p}).
\end{equation} By expressing the $a_p$'s in terms of the operators $b_p, b^{\dagger}_{-p}$, we see that it equals 
\begin{equation}\label{Fro2}
\mathbb{H}^{\text{F}}=\frac{-\triangle_R}{2M}+\sum_{p\neq 0}e_p b^{\dagger}_p b_p+\sum_{p\neq 0}\frac{|p|w_p}{\sqrt{e_p}}e^{-ipR}(b^{\dagger}_p+b_{-p})+E^{\rm B}
\end{equation} 
 which belongs to the class of Hamiltonians defined in \eqref{Fro1}. The Hamiltonian $\mathbb{H}^{\text{F}}$ acts on $L^2(\mathbb{T}^d)\otimes \mathcal{F}_+$, where $\mathcal{F}_+$ is the Fock space over the complement of the normalized constant function in $L^2(\mathbb{T}^d)$, describing solely the $p\neq 0$ modes of the field. 
In order to  obtain \eqref{Fro2} via a Bogoliubov approximation, we supplemented this procedure by additionally dropping, in the impurity-boson interaction, all the terms that are of higher order than \emph{linear} in the creation and annihilation operators (after first replacing the $a_0$ and its adjoint by $\sqrt{N}$), whereas we kept the quadratic terms in the boson-boson interaction. One of elements of our analysis below is the justification of this additional step while checking that the other steps, known to be rigorously justifiable in the mean-field case in the absence of an impurity, are still applicable. It is important, however, to realize that in some instances, especially when the impurity-boson interaction is strong, additional terms not present in the Fr\"ohlich Hamiltonian \eqref{Fro2} cannot be neglected \cite{Timour,Salmhofer,Enss}.

\subsection{Main results} 
The interpretation of our main results, as stated below, is that the Fr\"ohlich Hamiltonian \eqref{Fro2} may indeed be seen as an effective low-energy, large $N$ theory for the original model described by $H_N$ in \eqref{ham}.  Our analysis  consists of a rigorous justification of the extended Bogoliubov approximation, based on suitable operator inequalities. It leads to two main theorems, the first of which concerns the excitation spectrum of $H_N$.

\subsubsection{Theorem 1: convergence of eigenvalues} Let us denote by $e_i(A)$ the $i-$th eigenvalue resp. the $i-$th min-max value of an operator $A$, starting at $i=0$.  Our first Theorem states  that as long as one considers the energy levels of $H_N$ lying in a not too large window above the ground state, their values are provided by the corresponding eigenvalues of the Fr\"ohlich Hamiltonian if $N$ is sufficiently large. In particular, we provide explicit bounds on the size of that window as compared with $N$.

\begin{thm}\label{thm1}
Let $H_N$ and $\mathbb{H}^{\rm{F}}$ be defined by Eqs. \eqref{ham} and \eqref{Fro2}, respectively, and let $E_{\rm H}(N):=\frac{N}{2}v_0+\sqrt{N}w_0$. Assume that $v$ and $w$ satisfy Assumption~\ref{assum}. 
Then for all eigenvalues $e_i(H_N)$ such that  $e_i(H_N)-e_0(H_N)\leq \xi$ for some $\xi\geq 1$ we have  
\begin{equation}
 |e_i(H_N) -E_{\rm H}(N)-e_i(\mathbb{H}^{\text{F}})|\leq C_{v,w}\xi \left(\frac{\xi}{N}\right)^{1/2}
\end{equation} 
for some constant $C_{v,w}>0$ independent of the parameters $\xi$ and $N$. 
\end{thm}

\begin{rem}\label{rem1}
In the special case of the ground state energy we have 
\begin{equation}
\inf \mathrm{spec}\, H_N = \frac{1}{2}Nv_0+\sqrt{N}w_0+\inf \mathrm{spec}\, \mathbb{H}^{\rm{F}}+O(N^{-1/2}) . 
\end{equation}
The interaction with the impurity thus gives rise to a $N^{1/2}$ contribution to the ground state energy and, more importantly, leads to an $O(1)$ contribution to the excitation spectrum via the last term in \eqref{Fro3}. This can be understood as follows. In the impurity-free case, the effect of the emergence of phonons is reflected as a $O(1)$ correction to the ground state and low-lying excitation energies, in the mean-field limit considered here. There are only finitely many (even for large $N$) phonons that emerge in the system. The Fr\"ohlich model describes the impurity creating and annihilating excitations of the background. The number of the latter being $O(1)$, we expect that this phonon-impurity interaction should as well give rise to an $O(1)$ correction. The Bogoliubov approximation suggests that this interaction should scale as $\mu N^{1/2}$, hence we see that $\mu\sim 1/\sqrt{N}$ is consistent with these considerations. 
\end{rem}

\begin{rem}
\rm{The error bounds are of the form $\xi  (\xi/N)^{1/2}$. Therefore, as long as the total excitation energy satisfies $\xi\ll N$, the error made by using the Fr\"ohlich Hamiltonian instead of the original one when computing the energy levels is small compared to the total excitation energy. The size of this energy window is presumably optimal.  In fact, if the condition $\xi\ll N$  is not fulfilled one cannot expect the onset of BEC anymore, which is an essential assumption in the Bogoliubov approximation. 
It is noteworthy that precisely the same error scaling was obtained in \cite{Se11} for the pure bosonic system. The effects of the inclusion of the impurity thus manifest themselves only in the value of the constant $C_{v,w}$. }
\end{rem}

\begin{rem} By a direct inspection of the proof, one sees that the result can easily be generalized to the case of multiple impurities (as long as their number is fixed, i.e., independent of $N$). This holds irrespectively of the statistics of the impurities, i.e. they could be fermions, bosons, or distinguishable (in particular, different) particles.
\end{rem}

\begin{rem}
  

Extending the results to the case of more realistic, short-ranged potentials remains a challenge.  In fact, the $w\equiv 0$ cases with either $\lambda=N^{2/d}$, $\eta=N^{1/d}$ (equivalent to the thermodynamic limit) or $\lambda=N^2$, $\eta=N $ in $d=3$ (the \emph{Gross-Pitaevskii} limit) were  rigorously analyzed only very recently. The results for the thermodynamic limit concern the ground state energy only \cite{LiebYngvason, ErdoesSchleinYau, FournaisSolovej, YanYin}, whereas in the Gross-Pitaevskii scaling regime the emergence of the Bogoliubov spectrum for low energies was shown as well \cite{BoccatoBrenneckeSchleinCenatiempo}. 
\end{rem}

\begin{rem} If a contact interaction is used to model both boson-boson and boson-impurity interaction, one encounters the \emph{Bogoliubov--Fr\"ohlich Hamiltonian} \cite{Grusdt, Lampart}
\begin{equation}\label{HF}
\mathbb{H}^{\text{B-F}}=\frac{P^2}{2M}+\sum_p \epsilon_p b^{\dagger}_p b_p +\sqrt{n_0} g_{IB}\sum_p \left(\frac{(\zeta p)^2}{2+(\zeta p)^2}\right)^{1/4}(b^{\dagger}_p+b_{-p})e^{-ipR}, 
\end{equation} where $n_0$ is the condensate density and $\zeta=(2g_{BB}n_0)^{-1/2}$ is the {healing length}; the parameters $g_{IB}$ and $g_{BB}$ are the coupling constants describing the impurity-boson and boson-boson interactions, respectively. Additionally, $\epsilon_p=\sqrt{c^2p^2(1+(\zeta p)^2)}$ with $c=1/\zeta=\sqrt{2g_{BB}n_0}$ denoting the speed of sound in the bosonic bath. This Hamiltonian displays an evident ultraviolet divergence, recently analyzed in \cite{Lampart}. By naively replacing $v_p$ and $w_p$ in \eqref{Fro2} with the respective coupling constants $g_{BB}$ and $g_{IB}$, one arrives at $\mathbb{H}^{\text{B-F}}$ with unit condensate density. 
We conjecture that \eqref{HF}, resp. some renormalized version of it,  arises in place of $\mathbb{H}^\text{F}$ in scaling regimes corresponding to more realistic interactions of shorter range than the mean-field limit considered here.
\end{rem}

\begin{rem}
Our proof makes use of methods from \cite{Se11} and \cite{LNSS}.  In particular, in the case $w\equiv 0$, we reproduce the results of \cite{Se11}, but by utilizing techniques from \cite{LNSS} we are able to substantially simplify the proof. 
\end{rem}

\subsubsection{Theorem 2: convergence of eigenvectors}

In order to compare the two operators $H_N$ and  $\mathbb{H}^\text{F}$, which act on  different Hilbert spaces, 
we utilize 
an operator introduced by Lewin, Nam, Serfaty and Solovej in \cite{LNSS}, which maps $\mathcal{H}_N$ to (a subspace of) $\mathcal{F}_+$. We give here a quick review of their construction, as it is important to formulate our second result. 

\subsubsection{The LNSS transform}\label{ss:LNSS}

If $\lbrace v_i \rbrace_{i\geq 0}$ is an orthonormal basis of some Hilbert space $\mathcal{H}$, then the $N$-fold symmetric tensor product of $\mathcal{H}$ is spanned by $N$-fold tensor products
\begin{displaymath}v_{i_1}\otimes_{\textrm{s}} \cdots \otimes_{\textrm{s}} v_{i_N}:=\mathcal{N}\sum_{\sigma \in S_N}v_{\sigma(i_1)} \otimes \cdots \otimes v_{\sigma(i_N)}\end{displaymath} for all choices of indices $i_j \in \mathbb{N}\cup \lbrace 0 \rbrace$ with $\mathcal{N}$ a normalization constant. Let us fix an element $v_0$ in the basis of $\mathcal{H}$. If one defines $\mathcal{H}^0_l$ to be the span of 
$\bigotimes_{\textrm{s}}^l v_0\otimes_{\textrm{s}} v_{i_{l+1}}\otimes_{\textrm{s}} \cdots v_{i_{N}}$ 
for all choices of the $N-l$ indices $ i_j\neq 0$, it is clear that  
\begin{displaymath}
    \mathcal{H}_N=\bigoplus_{l=0}^N \mathcal{H}^0_l.
\end{displaymath}
For convenience, we further define $\mathcal{H}^+_m$ by the relation $\mathcal{H}^0_{N-m}=\lbrace \otimes_{\rm{s}}^{N-m}v_0\rbrace \otimes_{\rm{s}} \mathcal{H}^+_m$. Explicitly, \begin{displaymath}
  \mathcal{H}^+_m=\sideset{}{_{\text{s}}^m}\bigotimes \mathcal{H}^+, \quad \mathcal{H}^+:=\lbrace {v_0}\rbrace^{\perp} .
\end{displaymath}For every element $\Psi \in \mathcal{H}_N$, define the linear operator \begin{equation*}U: \mathcal{H}_N\rightarrow \mathcal{F_+^{\leq N}}, \quad \Psi \mapsto \phi_0\oplus \cdots \phi_N\end{equation*} where the $\phi_i \in \mathcal{H}^+_i$, $i\in \{0,\dots, N\}$, are uniquely determined by the above considerations. The space $\mathcal{F}^{\leq N}_+$ is naturally seen to be a proper subset of the Fock space over the orthogonal complement of $v_0 \in \mathcal{H}$. Moreover, $U$ is unitary. Performing this construction for $\mathcal{H}=L^2(\mathbb{T}^d)$ with, for instance, the plane wave basis and with $v_0\equiv 1$  we arrive at a unitary transformation $U: \mathcal{H}_N\rightarrow \mathcal{F}^{\leq N}_+\subset\mathcal{F}_+$ with $\mathcal{F}_+$ being the Fock space over the orthogonal complement of the unit function on $\mathbb{T}^d$. This space has a clear physical interpretation of being the  space of excitations from the condensate, and the fully condensed state plays the role of the vacuum. It is due to the algebraic properties of $U$, however, that it becomes helpful in the analysis, as it can be seen to rigorously realize the Bogoliubov substitution of $a_0, a^{\dagger}_0$ by $\sqrt{N}$. More precisely, with $\mathcal{Q}$ denoting the projection onto the orthogonal complement of the unit function in $L^2(\mathbb{T}^d)$, one can check that (the annihilation operator is here understood to be the standard operator in the purely bosonic Fock space) \begin{equation}
U(\Psi)=\bigoplus_{j=0}^N \mathcal{Q}^{\otimes_{\textrm{s}} j} \left(\frac{a_0^{N-j}}{\sqrt{(N-j)!}}\Psi\right) 
\end{equation} for all $\Psi \in \mathcal{H}_N$  and consequently that for $k,l\neq 0$
\begin{eqnarray} 
U^{\dagger}a_k^{\dagger}a_0 U=a_k^{\dagger}\sqrt{N-N_+}\label{wp}\\
U^{\dagger}a_k^{\dagger}a_l a^{\dagger}_0a_0 U=a^{\dagger}_ka_l (N-N_+)\\
U^{\dagger}a_k^{\dagger}a_l^{\dagger}a_0 a_0 U=a_k^{\dagger} a_l^{\dagger}\sqrt{(N-N_+)(N-N_+-1)}\label{vp}.
\end{eqnarray} The last two identities follow from the first, in fact. 
We trivially  extend this transformation to an operator $L^2(\mathbb{T}^d)\otimes \mathcal{H}_N \rightarrow L^2(\mathbb{T}^d)\otimes \mathcal{F}_+^{\leq N}$ by tensor-multiplying it by the unit operator on the impurity Hilbert space. This extended $U$ is again unitary and satisfies \eqref{wp} with $a_0$ defined by \eqref{animp}. One should keep in mind that $U$ depends on $N$. Equipped with the extended operator $U$,
we now state our second main result concerning the eigenvectors.

\begin{thm}\label{thm2}
Let  $\mathbb{P}_i$ denote the orthogonal projection onto the eigenspace of $\mathbb{H}^{\text{F}}$ corresponding to energy $e_i(\mathbb{H}^\text{F})$.
Under Assumption~\ref{assum}, the following statements hold true.
\begin{enumerate}
\item  The spectra of both $H_N$ and $\mathbb{H}^{\text{F}}$ are discrete.
\item  For all $i$  such that there exists an eigenstate $\Psi_i$ of $H_N$ corresponding to energy $e_i(H_N)$ with $e_i(H_N)-e_0(H_N)< \xi$ where $\xi>0$ is fixed, we have
\begin{equation}
\lim_{N\rightarrow \infty} (\Psi_i, U^{\dagger} \mathbb{P}_i U \Psi_i)_{L^2(\mathbb{T}^d)\otimes \mathcal{F}_+}=1.
\end{equation}
\end{enumerate}
\end{thm}

\begin{rem}
In contrast to the case without impurity, the eigenstates of $\mathbb{H}^{\rm{F}}$ are not explicit. In particular, they display non-trivial correlations among the phonons and are not quasi-free. 
\end{rem}

\begin{rem}
We have not tried to find the rate of growth of the size of the energy window in $N$ so as provide the corresponding error for replacing eigenvectors. This rate is probably much worse than the one from Theorem~\ref{thm1}. 
\end{rem}

\begin{rem}
Theorems~\ref{thm1} and~\ref{thm2} together imply, as $N\rightarrow \infty$, the norm resolvent convergence of $H_N-E_{\rm H}(N)$ towards $\mathbb{H}^{\text{F}}$, that is, for any $z\in \mathbb{C} \backslash \mathbb{R}$, \begin{equation}\label{reso}
\lim_{N\rightarrow\infty}||(U_N(H_N-E_{\rm H}(N))U_N^{\dagger}-z)^{-1}-(\mathbb{H}^\text{F}-z)^{-1}||=0
\end{equation} in operator norm. Here $U_N$ has to be understood as a partial isometry, i.e., $U_N^\dagger$ is extended by $0$ to all of $L^2(\mathbb{T}^d)\otimes\mathcal{F}_+$. 
\end{rem}

\begin{rem}
Another interesting  problem concerns the dynamics of the impurity and the use of the Fr\"ohlich Hamiltonian as its generator. This question has been recently studied from a  physics perspective \cite{Salmhofer, Lewenstein}. From a mathematical point of view, there exist results concerning the dynamics of a tracer particle immersed in a Bose gas \cite{tracer,tracer2}, which concern a different scaling limit than the one considered here and do not utilize the Fr\"ohlich description. The convergence \eqref{reso} can also be reformulated as convergence of the corresponding group of time evolutions, and hence can be used to determine also the dynamics of small excitations of the condensate.   In the absence of an impurity, more general results are known where the condensate itself is excited and evolves according to the time-dependent Hartree equation  (see, e.g., \cite{Schlein,NamMarcin}).
\end{rem}

The remainder of this paper contains the proofs of Theorems \ref{thm1} and \ref{thm2}.   Throughout the text, the symbol $C$ denotes a positive constant whose value may change at different appearances. Moreover, unless stated otherwise, all states on the relevant Hilbert spaces are normalized. Finally, all operators that are defined as acting on functions of the Bose gas coordinates or the field modes only are actually everywhere understood as their tensor products with the unit operator on $L^2(\mathbb{T}^d)$, the latter being the Hilbert space of the impurity particle. 

\section{Auxiliary considerations}

In this Section we introduce four preparatory Lemmas that will be needed in the proofs of Theorems~\ref{thm1} and~\ref{thm2}. 
For their statement, we need to introduce some notation.  We shall often denote the terms on the right side of \eqref{ham},  from left to right, by $P^2/2M, T, V$ and $W$. Let $\mathcal{P}$ denote the projection onto the normalized constant wave function in $L^2(\mathbb{T}^d)$, and $\mathcal{Q}=1-\mathcal{P}$. We define the excitation number operator
\begin{equation}
N_+=\sum_{i=1}^N\mathcal{Q}_i
\end{equation} as an operator on $\mathcal{H}_N$. The sub-index in $Q_i$ means here that we project onto the orthogonal complement of the normalized constant wave function in the $i$-th variable. The second quantized form of the excitation number operator in the plane wave basis equals 
\begin{equation}
N_+=\sum_{p\neq 0}a_p^{\dagger}a_p.
\end{equation}
The first Lemma explores the consequences of the mean-field structure of $H_N$. In particular, the ground state energy of $H_N$ is, to leading order in $N$, equal to $E_{\rm H}(N)$, and the excitation number operator is uniformly bounded in $N$ for states of fixed excitation energy.

\begin{lemma}\label{lemma1}
The ground state energy of $H_N$, $e_0(H_N)$, satisfies the bounds \begin{equation}\label{grstatebds}
  \frac{Nv_0}{2}+\sqrt{N}w_0\geq e_0(H_N) \geq  \frac{Nv_0}{2}+\sqrt{N}w_0 -\delta E
\end{equation} with $\delta E=\int (2\pi^2)^{-1 }w^2+(v(0)-v_0)\geq 0$. Moreover, we have the operator inequality \begin{equation}
N_+\leq C (H_N-e_0(H_N))+C.
\end{equation} 
\end{lemma}

\begin{rem}
Below, we will make use of a direct consequence of this Lemma, namely \begin{equation}
  (\Psi, N_+ \Psi)\leq C\xi+C
\end{equation} for any state $\Psi$ such that $(\Psi,H_N\Psi)\leq e_0(H_N)+\xi$ with $\xi>0$.
\end{rem}

\begin{proof}
The upper bound on the ground state energy is obtained by taking the constant wave function in $L^2(\mathbb{T}^d)\otimes \mathcal{H}_N$ as trial function. We write $H_N=\frac{P^2}{2M}+\frac{1}{2}T+V+(\frac{1}{2}T+W)$; by a standard argument using the positivity of the Fourier coefficients of $v$ we have 
\begin{align}\nonumber
V& =\frac{1}{2(N-1)}\sum_{i,j\in \lbrace{1,\dots, N \rbrace} }v(x_i-x_j)-\frac{Nv(0)}{2(N-1)}\\ \nonumber & =\frac{1}{2(N-1)}\sum_p v_p \left|\sum_{i=1}^Ne^{\mathrm{i}px_i}\right|^2-\frac{Nv(0)}{2(N-1)}\\ & \geq    \frac{N}{2}v_0-\frac{N}{2(N-1)}(v(0)-v_0)  \label{fouri}
\end{align} 
since $\sum_{p\neq 0}v_p \left|\sum_{i=1}^N e^{\mathrm{i}px_i}\right|^2\geq 0$.
Next, we use Temple's inequality, see, e.g.,\cite{Temple}. Consider a Hamiltonian $H=H_0+Z$ with non-negative self-adjoint operators $Z$ and  $H_0$  with ground state energy satisfying $e_0(H_0)=0$. Denoting by $e_0,e_1$ the first two eigenvalues of $H$, we have clearly $(H-e_0)(H-e_1)\geq 0$. We evaluate this at the ground state of $H_0$, $\Psi_0$.  We get \begin{equation*}
(\Psi_0,(H-e_0)(H-e_1) \Psi_0)= (\Psi_0,(Z-e_0)(Z-e_1)\Psi_0)\geq 0
\end{equation*} and rewrite this, since $e_1>0$, as 
\begin{equation}
e_0 \geq -\frac{(\Psi_0, Z^2\Psi_0)}{e_1}+\left(1+\frac{e_0}{e_1}\right)(\Psi_0, Z \Psi_0).
\end{equation} 
Using the positivity of $Z$ and  
$e_1\geq e_1(H_0)$  
we finally get 
\begin{equation}
e_0 \geq (\Psi_0, Z\Psi_0) -\frac{(\Psi_0 Z^2 \Psi_0)}{e_1(H_0)}.
\end{equation}   
Using this for $H=-\frac{\triangle_x}{2}+N^{-1/2}w(x-R)$ with $Z=N^{-1/2}w(x-R)$ and $\Psi_0$  the normalized constant function on $\mathbb{T}^d$, we have,
\begin{equation}\label{temple}
e_0\left(-\frac{\triangle_x}{2}+N^{-1/2}w(x-R)\right)\geq N^{-1/2}w_0-N^{-1}(2\pi^2)^{-1}\int w^2.
\end{equation} This leads to \begin{equation}
(\Psi,H_N-E_{\rm H}(N) \Psi) \geq (\Psi,\frac{T}{2}\Psi)  -\left(\frac{N}{2(N-1)}(v(0)-v_0)+(2\pi^2)^{-1}\int w^2\right)\|\Psi\|^2.
\end{equation} Using that $N_+\leq (2\pi)^{-2}T$, we see that the desired result holds.
\end{proof}

The second Lemma concerns the fluctuations of the condensate in the ground state, which are seen to be strongly suppressed due to the mean field scaling. 

\begin{lemma}\label{lemma2}
For all $N\geq 2$  we have the operator inequality \begin{equation}
N_+^2\leq C (H_N-e_0(H_N))^2+C.
\end{equation} 
\end{lemma}

\begin{rem} Similarly as above, the Lemma immediately implies that if $\Psi$ belongs to the spectral subspace of $H_N$ corresponding to energy $E\leq e_0(N)+\xi$ with $\xi\geq 0$, then we have 
\begin{equation}
(\Psi, N_+^2 \Psi)\leq C \xi^2+C
\end{equation} where the constants depend only on $v$ and $w$ but not on $N$. This will be of importance below.
\end{rem}

\begin{proof}
Because $N_+\leq \frac{1}{2\pi^2}(\frac{1}{2}T)$ and $N_+$ commutes with $T$, we find it convenient to give a bound on the operator $\frac 12 N_+ {T}$, as the latter can be directly linked to $H_N$. Writing \begin{equation}
\frac{T}{2}=(H_N-e_0(H_N))+S_1 + S
\end{equation} with \begin{equation}
S_1=-\frac{1}{N-1}\sum_{j=2}^N v(x_1-x_j)-\frac{(-\triangle_1)}{2}-\frac{w(x_1-R)}{\sqrt{N}}
\end{equation} and \begin{equation}
S=e_0(H_N)-\frac{1}{N-1}\sum_{2\leq i\leq j \leq N} v(x_i-x_j)-\frac{1}{\sqrt{N}}\sum_{j=2}^N w(x_j-R)-\sum_{j=2}^N \frac{-\triangle_j}{2}-\frac{P^2}{2M}
\end{equation} we estimate the relevant terms. By the Cauchy--Schwarz inequality, 
\begin{equation}
(\Psi, N_+(H_N-e_0(H_N)) \Psi)\leq \sqrt{(\Psi, N_+^2\Psi)} \sqrt{(\Psi, (H_N-e_0(H_N))^2 \Psi)}.
\end{equation} 
Note that $(S+S_1)\Psi$ is permutation symmetric in the Bose gas coordinates, so that $(\Psi, N_+ (S+S_1) \Psi)=N(\Psi, \mathcal{Q}_1 (S+S_1) \Psi)$, where $\mathcal{Q}_1=1-\mathcal{P}_1$. Moreover, $S$ is independent of $x_1$ hence it commutes with $\mathcal{Q}_1$. Using the inequality \eqref{fouri} (with $N$ replaced with $N-1$) as well as  Temple's inequality \eqref{temple} and  the upper bound on $e_0(H_N)$ in \eqref{grstatebds}, we see that \begin{equation*}  
S \leq  \frac{v_0 + v(0)}{2} + \frac{w_0}{\sqrt{N}} + \frac{N-1}{N}  \frac{\int w^2}{2\pi^2} =:\delta E' . 
\end{equation*}  
Since $S$ commutes with $\mathcal{Q}_1$ we thus have 
\begin{equation}\label{one}
N(\Psi \mathcal{Q}_1 S\Psi)\leq \delta E' (\Psi, N_+ \Psi) .
\end{equation} 
The part of $N_+ S_1$ not containing $-\triangle_1/2+N^{-1/2}w(x_1-R)$ is equal to $-N(\Psi,\mathcal{Q}_1 v(x_1-x_2)\Psi)$. We introduce the short-hand $v_{12}$ to denote $v(x_1-x_2)$. We write, following \cite{Se11} \begin{eqnarray}
(\Psi,\mathcal{Q}_1 v_{12}\Psi)=(\Psi, \mathcal{Q}_1 \mathcal{Q}_2 v_{12}\Psi)+(\Psi, \mathcal{Q}_1 \mathcal{P}_2 v_{12}\mathcal{P}_2 \Psi)+(\Psi, \mathcal{Q}_1 \mathcal{P}_2 v_{12} \mathcal{Q}_2 \Psi). 
\end{eqnarray} 
Observe that $(\Psi, \mathcal{Q}_1 \mathcal{P}_2 v_{12} \mathcal{P}_2 \Psi)=(\Psi, \mathcal{Q}_1 \mathcal{P}_2 v_{12} \mathcal{P}_2 \mathcal{Q}_1 \Psi)+(\Psi, \mathcal{Q}_1 \mathcal{P}_2 v_{12} \mathcal{P}_2 \mathcal{P}_1\Psi)$, where the last term vanishes and the remaining one is positive. 
For the first term, we use $(\Psi, \mathcal{Q}_1\mathcal{Q}_2 v_{12}\Psi)\geq -\|v\|_{\infty} \sqrt{(\Psi \mathcal{Q}_1 \mathcal{Q}_2\Psi)}$. Furthermore, 
\begin{align}\nonumber
(\Psi, \mathcal{Q}_1 \mathcal{P}_2 v_{12}\mathcal{Q}_2 \Psi) & \geq -\frac{1}{2}(\Psi,\mathcal{Q}_2 v_{12}\mathcal{Q}_2\Psi)-\frac{1}{2}(\Psi, \mathcal{Q}_1 \mathcal{P}_2 v_{12}\mathcal{P}_2 \mathcal{Q}_1\Psi) 
 \\ & \geq -\frac{\|v\|_{\infty}}{2}\left((\Psi, \mathcal{Q}_2\Psi)+(\Psi, \mathcal{Q}_1\mathcal{P}_2\mathcal{Q}_1\Psi)\right)\geq -\|v\|_{\infty}(\Psi, \mathcal{Q}_1\Psi)
\end{align} 
as $\mathcal{P}_2\leq 1$ and $(\Psi,\mathcal{Q}_1\Psi)=(\Psi, \mathcal{Q}_2\Psi)$ due to the permutation symmetry. The remaining part of $S_1$ is bounded as \begin{equation*}
  \left(\Psi, \mathcal{Q}_1 \left( \frac{-\triangle_1}{2}+\frac{1}{\sqrt{N}}w(x_1-R)\right) \Psi\right)\geq -\frac{\|w\|_{\infty}}{2} \left( \frac 1 N + (\Psi, \mathcal{Q}_1 \Psi) \right) 
\end{equation*} since $w\geq 0$. 

We thus have 
\begin{equation}
(\Psi,\mathcal{Q}_1 S_1 \Psi)\leq \|v\|_{\infty} \sqrt{(\Psi, \mathcal{Q}_1 \mathcal{Q}_2\Psi)}+(\|v\|_{\infty}+\tfrac{1}{2} \|w\|_\infty)(\Psi,\mathcal{Q}_1\Psi)+\frac{\|w\|_{\infty}}{2N}.
\end{equation} With $N^2(\Psi, \mathcal{Q}_1 \mathcal{Q}_2\Psi)\leq (\Psi, N_+^2\Psi)$, this altogether implies 
\begin{equation}\label{nn}
\frac{1}{2}(\Psi, N_+ T\Psi)\leq \left( \|v\|_\infty +\sqrt{(\Psi, (H_N-e_0(H_N)^2 \Psi)}\right)\sqrt{(\Psi, N_+^2\Psi)}+\alpha (\Psi, N_+ \Psi)+\frac{\|w\|_{\infty}}{2},
\end{equation} 
where the $N$-independent constant $\alpha$ equals $\alpha = \frac{1}{2}\|w\|_\infty +\|v\|_{\infty}+\delta E'$.  
As $N_+\leq g T$, with $g=(2\pi)^2$ being the energy gap of the Laplacian on the torus, this implies \begin{equation*}
  g N_+^2 \leq \frac{\|v\|_{\infty}^2}{\kappa} + \frac{ (H_N-e_0(H_N))^2}{\lambda}+\frac{\alpha^2}{\epsilon}+\|w\|_{\infty}+ (\kappa+\epsilon+\lambda) N_+^2
\end{equation*} for any $\epsilon, \lambda, \kappa >0$. By choosing $\epsilon=\lambda=\kappa=\frac{g}{4}$, we arrive at the desired result. 
\end{proof}

The third and fourth Lemmas concern $\mathbb{H}^{\rm{F}}$. They will be of importance when proving the upper bound on the difference of eigenvalues in Theorem \ref{thm1}. 

\begin{lemma}\label{lemma3}
Let $\mathbb{H}^{\rm{F}}_0=\frac{P^2}{2M}+\sum_{p\neq 0} (p^2+v_p)a_p^{\dagger}a_p$ denote the particle-conserving part of the Fr\"ohlich Hamiltonian \eqref{Fro3}. Then there exist positive constants $C_0,C_1, C_2$ such that the inequalities
\begin{equation}
N_+\leq C_0 \mathbb{H}^{\rm{F}}_0 \leq C_1 \mathbb{H}^{\rm{F}}+C_2
\end{equation} hold true on $L^2(\mathbb{T}^d)\otimes\mathcal{F}_+$.
\end{lemma}

\begin{proof}
Clearly, as $v_p\geq 0$, one can take $C_0=g^{-1}=(2\pi)^{-2}$. The particle non-conserving part of $\mathbb{H}^{\text{F}}$ consists of the purely bosonic ($v$-dependent) part $V^{\rm{OD}}$and a $w$-dependent part $\tilde{W}$. The latter can be bounded by \begin{equation}
\tilde{W}\geq -\epsilon \mathbb{H}_0^{\rm{F}} -\epsilon^{-1}\sum_{p\neq 0} \frac{|w_p|^2}{v_p+p^2}
\end{equation} for any $\epsilon>0$. To see this, simply complete the square for a single mode using the inequality $(\eta a_p^{\dagger}+\eta^{-1}w_p e^{ipR})(\eta a_p+\eta^{-1}w_p e^{-ipR})\geq 0$, then choose $\eta^2=\epsilon (p^2+v_p)$ and sum over the modes.  It is hence enough to show that the bosonic particle non-conserving part, given by \begin{equation}
V^{\rm{OD}}=\frac{1}{2}\sum_{p\neq 0} v_p(a^{\dagger}_pa^{\dagger}_{-p}+a_pa_{-p})
\end{equation} can be bounded below by $-c \mathbb{H}^{\rm{F}}_0 -c'$ for $0<c<1$ and $c'>0$. By Cauchy--Schwarz, \begin{equation}\label{csv}
\frac{v_p}{2}(a^{\dagger}_pa^{\dagger}_{-p}+a_pa_{-p})\geq -\epsilon a^{\dagger}_p a_p-\frac{|v_p|^2}{4\epsilon}a^{\dagger}_{-p}a_{-p}-\frac{|v_p|^2}{4\epsilon}
\end{equation} for any $\epsilon>0$. Now take $\epsilon=\lambda (p^2+v_p)$ for some $\lambda>0$ and 
define $\mu:=\frac{\sup_{p\neq 0} v_p^2}{\sup_{p\neq 0}v_p^2+\inf_{p\neq 0} p^2(p^2+2v_p)}$; then $0<\mu<1$ (recall that $p\in (2\pi \mathbb{Z})^d)$ and $v_p^2 \leq \frac{\mu}{1-\mu}p^2(p^2+2v_p)$, or 
\begin{equation}
\frac{v_p^2}{p^2+v_p}\leq \mu (p^2+v_p). 
\end{equation} Consequently, \begin{equation}\label{vod}
V^{\rm{OD}}\geq -(\lambda+\frac{\mu}{4\lambda})\mathbb{H}_0^{\rm{F}}-\sum_p \frac{v_p^2}{\lambda(p^2+v_p)}.
\end{equation} By choosing $\lambda=\frac{\sqrt{\mu}}{2}$, we have $\lambda+\frac{\mu}{4\lambda}=\sqrt{\mu}<1$ and the desired result follows.
\end{proof}

\begin{rem}
Note that the above Lemma implies that $\mathbb{H}^{\rm{F}}$ is bounded from below. 
\end{rem}

The last Lemma relates $N_+^2$ to $(\mathbb{H}^{\rm{F}})^2$. 

\begin{lemma}\label{lemma4}
On $L^2(\mathbb{T}^d)\otimes\mathcal{F}_+$ we have 
\begin{equation}
  N_+^2 \leq C (\mathbb{H}^{\rm{F}})  ^2 +C.
\end{equation}
\end{lemma}

\begin{proof} 
We will show that $N_+ \mathbb{H}_0^{\rm{F}}\leq C (\mathbb{H}^{\rm{F}})^2+C$, which  implies the desired result by the previous lemma. 
As $[N_+, \mathbb{H}^{\rm{F}}_0]=0$, we have \begin{equation}
 N_+\mathbb{H}_0^{\rm{F}}=\frac{1}{2}(N_+\mathbb{H}_0^{\rm{F}}+\mathbb{H}_0^{\rm{F}}N_+)=\frac{1}{2}(N_+ \mathbb{H}^{\rm{F}}+\mathbb{H}^{\rm{F}}N_+)-\frac{1}{2}(N_+ V^{\rm{OD}}+V^{\rm{OD}}N_++\tilde{W}N_++N_+\tilde{W}),
\end{equation} 
with $\tilde W$ and $V^{\rm OD}$ defined as in the proof of Lemma~\ref{lemma3}. 
Using the canonical commutation relations $[a_p, a_q^{\dagger}]=\delta_{p,q}$, we compute 
\begin{equation}
 N_+V^{\rm{OD}} = \sum_{p\neq 0} a^{\dagger}_p V^{\rm{OD}} a_p + \sum_{p\neq 0}\frac{v_p}{2}a^{\dagger}_p a^{\dagger}_{-p} .
 \end{equation}
Since $V^{\rm{OD}}N_+=(N_+ V^{\rm{OD}})^{\dagger}$
we have $V^{\rm{OD}}N_+=\sum_{p\neq 0 }a^{\dagger}_pV^{\rm{OD}}a_p + \sum_{p\neq 0} \frac{v_p}{2}a_p a_{-p}$ and finally \begin{equation}\label{computation}
\frac{1}{2}(N_+V^{\rm{OD}}+V^{\rm{OD}}N_+)=\sum_{p\neq 0} a^{\dagger}_p V^{\rm{OD}}a_p +\frac{1}{2}V^{\rm{OD}}.
\end{equation}


Using \eqref{vod} and  the fact that $\sum_{p\neq 0}a^{\dagger}_p \mathbb{H}_0^{\rm{F}}a_p=\mathbb{H}_0^{\rm{F}}(N_+-1)$,  we have \begin{equation}
-\frac{1}{2}(N_+V^{\rm{OD}}+V^{\rm{OD}}N_+)\leq \sqrt{\mu} \mathbb{H}^{\rm{F}}_0 N_+ +\frac{\sqrt{\mu}}{2}\mathbb{H}^{F}_0+C
\end{equation} where $\mu<1$. By Lemma 2.3 and the Cauchy--Schwarz inequality, the last two terms of the above are bounded by $C(\mathbb{H}^{\rm{F}})^2+C$. 

For $\tilde{W}$ we perform a computation analogous to \eqref{computation}, which yields 
\begin{equation}
\frac{1}{2}(N_+ \tilde{W}+\tilde{W}N_+)=\sum_{p\neq 0} a^{\dagger}_p \tilde{W}a_p+\frac{1}{2}\tilde{W}.
\end{equation} By completing the square similarly as in Lemma 2.3, we have \begin{equation}
\tilde{W}\geq -\lambda N_+ -\frac{1}{\lambda}\sum_{p\neq 0}|w_p|^2
\end{equation} 
for any $\lambda>0$. We obtain 
\begin{equation}
-\frac{1}{2}(N_+ \tilde{W}+\tilde{W}N_+)\leq \lambda N_+(N_+-1) +\frac{\sum_{p\neq 0}|w_p|^2}{\lambda}N_++\frac{1}{2}N_++\frac{\sum_{p\neq 0}|w_p|^2}{2}
\end{equation} 
for any $\lambda>0$. By Lemma~\ref{lemma3} and $\mathbb{H}^{\rm{F}}\leq \frac{(\mathbb{H}^{\rm{F}})^2}{2}+\frac{1}{2}$, we can bound \begin{equation}
-\frac{1}{2}(N_+ \tilde{W}+\tilde{W}N_+) \leq \lambda C_0 N_+ \mathbb{H}^{\rm{F}}_0 + (\frac{1}{\lambda}+1)C(\mathbb{H}^{\rm{F}})^2+C.
\end{equation}

Finally, using again the Cauchy--Schwarz inequality, we can bound \begin{equation}
N_+\mathbb{H}^{\rm{F}}+\mathbb{H}^{\text{F}}N_+ \leq \epsilon N_+^2 +\frac{1}{\epsilon}(\mathbb{H}^{\rm{F}})^2
\end{equation} for any $\epsilon>0$. Invoking Lemma 2.3 again, we obtain for any $\epsilon>0$ and $\lambda>0$, \begin{equation}
\left(1-\sqrt{\mu}-\frac{1}{2}C_0(\epsilon+2\lambda)\right)N_+\mathbb{H}^{\rm{F}}_0 \leq ((2\epsilon)^{-1}+\lambda^{-1})C (\mathbb{H}^{\rm{F}})^2+C.
\end{equation} By choosing $\epsilon$ and $\lambda$ small enough, we arrive at the desired result. 
\end{proof}

\section{Comparing $H_N$ and $\mathbb{H}^{\rm{F}}$}

The estimates provided in the previous section concern the relation of the number of excitations operator $N_+$ (or its square) to the Hamiltonians $H_N$ and $\mathbb{H}^{\text{F}}$ independently. Now, making use of the LNSS transformation $U$ introduced in Sec.~\ref{ss:LNSS}, we give an important estimate relating $UH_NU^{\dagger}$ and $\mathbb{H}^{\text{F}}$.

\begin{proposition}\label{prop}
There exist positive constants $\alpha, \beta$, independent of $N$, such that for every $\epsilon>0$ and every $\Phi$ in $L^2(\mathbb{T}^d)\otimes \mathcal{F}_+^{\leq N}$ we have the inequality \begin{equation}
\left|\left(\Phi,\left(U(H_N-E_{\rm H}(N))U^{\dagger}-\mathbb{H}^{\rm{F}}\right)\Phi\right)\right|\leq \alpha \frac{(\Phi,N_+^2\Phi)}{N}\left(1+\frac{1}{\epsilon}\right)+\beta(\Phi,N_+ \Phi) \left(\epsilon+\frac{1}{\sqrt{N}}\right).
\end{equation}
\end{proposition}

The proof of the proposition is divided into two main steps. In step 1, we take care of the higher-order terms in the creation and annihilation operators that appear in the second quantization of $H_N$, but are absent in $\mathbb{H}^{\rm{F}}$. Let 
\begin{align}\nonumber
H^{\rm{pre-F}}_N &:=\frac{P^2}{2M}+\sum_{p\neq 0} p^2 a^{\dagger}_p a_p +\frac{1}{2(N-1)}\sum_{p\neq 0} v_p (2a^{\dagger}_p a_pa_0^{\dagger}a_0+a_p^{\dagger}a_0 a_0 a_{-p}^{\dagger}+a_pa_0^{\dagger} a_0^{\dagger}a_{-p})\\ & \qquad +\frac{1}{\sqrt{N}}\sum_{p\neq 0}w_p e^{-ipR}(a_p^{\dagger}a_0+a_{-p} a_0^{\dagger}).
\end{align}
viewed as an operator on $L^2(\mathbb{T}^d)\otimes\mathcal{H}_N$. 

\begin{lemma}\label{higherorder}
For any $\epsilon>0$, one has the operator inequalities \begin{equation}
-E_{\epsilon} \leq H_N-E_{\rm H}(N)-H^{\rm{pre-F}}_N \leq F_{\epsilon}
\end{equation} where \begin{equation}E_{\epsilon}=\frac{N_+(N_+-1)}{2(N-1)}\left(v_0+\frac{v(0)}{\epsilon}\right)+\epsilon v_0 \frac{2N-1}{N-1}N_+\end{equation} and \begin{equation}F_{\epsilon}=\frac{\|w\|_{\infty}}{\sqrt{N}}N_++\epsilon v_0 \frac{2N-1}{N-1}N_++\left(1+\frac{1}{\epsilon}\right)\frac{N_+(N_+-1)}{2(N-1)}v(0).\end{equation}
\end{lemma}

\begin{proof} 
Using the Cauchy--Schwarz inequality and positivity of $v$ viewed as a two-particle multiplication operator, we have
\begin{align}\nonumber
& \pm \left(  (\mathcal{P}\otimes \mathcal{Q}+\mathcal{Q}\otimes \mathcal{P})v(\mathcal{Q} \otimes \mathcal{Q}) + (\mathcal{Q} \otimes \mathcal{Q})v (\mathcal{P}\otimes \mathcal{Q}+\mathcal{Q}\otimes \mathcal{P})  \right)   \\ & \leq  \epsilon (\mathcal{P}\otimes\mathcal{Q} +\mathcal{Q}\otimes \mathcal{P}) v (\mathcal{P}\otimes \mathcal{Q} +\mathcal{Q}\otimes \mathcal{P}) +\frac{1}{\epsilon}(\mathcal{Q}\otimes \mathcal{Q} )v(\mathcal{Q}\otimes \mathcal{Q}).
\end{align} 
By translation invariance $Q\otimes \mathcal{P} v \mathcal{P}\otimes \mathcal{P} =0$. 
Moreover, the boundedness of $v$ enables us to bound 
\begin{equation}
\mathcal{Q}\otimes \mathcal{Q} v \mathcal{Q} \otimes \mathcal{Q} \leq v(0) \mathcal{Q}\otimes \mathcal{Q} .
\end{equation} 
Therefore, we have the bounds 
\begin{equation}\label{bd1}
\begin{split}
v & \geq \mathcal{P}\otimes \mathcal{P} v \mathcal{P}\otimes \mathcal{P} +\mathcal{P}\otimes \mathcal{P} v \mathcal{Q} \otimes \mathcal{Q} +\mathcal{Q}\otimes \mathcal{Q} v \mathcal{P}\otimes   \mathcal{P}  \\&  \quad +(1-\epsilon)\left(\mathcal{P}\otimes \mathcal{Q} +\mathcal{Q}\otimes \mathcal{P}\right)v\left(\mathcal{P}\otimes \mathcal{Q}+\mathcal{Q}\otimes \mathcal{P}\right)-\epsilon^{-1}v(0)\mathcal{Q}\otimes \mathcal{Q}  \end{split}\end{equation} and \begin{equation}\begin{split}
v &\leq \mathcal{P}\otimes \mathcal{P} v \mathcal{P}\otimes \mathcal{P} +\mathcal{P}\otimes \mathcal{P} v \mathcal{Q} \otimes \mathcal{Q} +\mathcal{Q}\otimes \mathcal{Q} v \mathcal{P}\otimes \mathcal{P}  \\ & \quad +(1+\epsilon)\left(\mathcal{P}\otimes \mathcal{Q} +\mathcal{Q}\otimes \mathcal{P}\right)v\left(\mathcal{P}\otimes \mathcal{Q}+\mathcal{Q}\otimes \mathcal{P}\right)+(1+\epsilon^{-1})v(0)\mathcal{Q}\otimes \mathcal{Q}.  \end{split}
\end{equation}
Similarly, treating $w(x-R)$ as a {one-body} multiplication operator parametrized by $R$, we have
\begin{equation}\label{bd2}
0\leq w \leq \mathcal{P}w\mathcal{P}+\mathcal{Q}w\mathcal{P}+\mathcal{P}w\mathcal{Q} +\|w\|_{\infty}\mathcal{Q}.
\end{equation}  
Taking into account that \begin{equation}\label{later}
(N-1)^{-1}\sum_{p\neq 0} v_p a^{\dagger}_p a_0^{\dagger}a_0 a_p \leq v_0 N_+
\end{equation} one easily arrives, after computing the relevant second quantization representations of the operators appearing in the bounds  \eqref{bd1} and \eqref{bd2}, at the desired result. Since this is essentially the same computation as in \cite[Sec. 5]{Se11}, we omit the details.
\end{proof}

The operator inequalities in Lemma \ref{higherorder} quantify the effect of dropping the higher order terms in the creation and annihilation operators appearing in the original Hamiltonian. As a second step,  we  now  estimate the effect of the Bogoliubov substitution of $a_0, a_0^{\dagger}$ by $ \sqrt{N}\in \mathbb{R}$ via the unitary transform $U$, which replaces the $a_0, a_0^{\dagger}$ by an {operator} $\sqrt{N-N_+}$ acting on $\mathcal{F}^{\leq N}_+$. 

\begin{lemma}\label{preFvsH}
We have the following inequality for all $\Phi\in L^2(\mathbb{T}^d)\otimes \mathcal{F}^{\leq N}_+$:
\begin{equation}
|(\Phi,U H_N^{\rm{pre-F}}U^{\dagger}-\mathbb{H}^{\rm{F}},\Phi)|\leq \frac{\alpha'(\Phi, N_+^2\Phi)+ \beta' \| \Phi\|^2}{(N-1)},
\end{equation} where the positive constants $\alpha', \beta'$ do not depend on $N$. 
\end{lemma}

\begin{proof} By using the algebraic properties \eqref{wp}--\eqref{vp} of $U$ we see that the expressions to estimate are the following. First, using \eqref{wp},
\begin{align}\nonumber
&|(\Phi, \left[ N^{-1/2}\sum_{p\neq 0} w_p e^{-ipR} U(a^{\dagger}_pa_0+a_p a^{\dagger}_0)U^{\dagger}-\sum_{p\neq 0} w_pe^{-ipR} (a^{\dagger}_p+a_{-p})\right] \Phi)| \\ &=  \nonumber
\left|\sum_{p\neq 0} (\Phi,w_p \left (a_p^{\dagger}e^{-ipR}\left(1-\sqrt{\frac{N-N_+}{N}}\right)+\left(1-\sqrt{\frac{N-N_+}{N}}\right)a_{-p}e^{ipR}\right) \Phi)\right|  \\ & \leq \epsilon^{-1} \frac{(\Phi,N_+^2\Phi)}{N^2}\sum_p |w_p|^2+ \epsilon (\Phi,N_+\Phi)
\end{align} 
which gives an expression of the type claimed Proposition \ref{prop} for $\epsilon^{-1}=N^2/(N-1)$. In the above, we used the Cauchy--Schwarz inequality \begin{equation}\label{extcs}
 AB+BA^{\dagger}\leq \epsilon A^{\dagger}A+\epsilon^{-1}B^2 
\end{equation}
for $A=a^{\dagger}_p e^{-ipR}$ and $B=w_p(1-\sqrt{(N-N_+)/N})$, and used the bound \begin{equation*}B^2=w_p^2N^{-1}(\sqrt{N}+\sqrt{N-N_+})^{-2}N_+^2\leq w_p^2 N_+^2/N^2. \end{equation*} 
Similarly, from \eqref{vp}, we arrive at the second term to estimate:  
\begin{align}\nonumber
&  \left| \sum_{p\neq 0}  (\Phi, (v_p a^{\dagger}_p a^{\dagger}_{-p}  \left(\frac{\sqrt{(N-N_+)(N-N_+-1)} }{N-1}-1\right)+h.c.)\Phi)\right|   \\    \nonumber
& \leq \epsilon^{-1} \sum_{p\neq 0} |v_p|^2 \frac{(\Phi,(N_++1)^2\Phi)}{(N-1)^2}+ \sum_{p\neq 0} \epsilon (\Phi,a^{\dagger}_p a^{\dagger}_{-p}a_{-p}a_p\Phi)    \\ 
& \leq  C\frac{(\Phi,(N_++1)^2\Phi)}{N-1}+\frac{(\Phi,N_+(N_+-1)\Phi)}{N-1}
\end{align} 
for $\epsilon^{-1}= N-1$. We used \eqref{extcs} for $A=a^{\dagger}_p a^{\dagger}_{-p}$ and 
\begin{equation*} 
B= v_p \left( \frac{\sqrt{(N-N_+)(N-N_+-1)}}{N-1}-1 \right) ,
\end{equation*} 
whose square is bounded by $v_p^2(\frac{N_++1}{N-1})^2$. Additionally, 
\begin{equation}\label{aps}
\sum_{p\neq 0} a^{\dagger}_p a^{\dagger}_{-p}a_{-p}a_p
\leq \sum_{p\neq 0}a^{\dagger}_pN_+ a_p=N_+^2-N_+ .
\end{equation}  Similarly, 
\begin{align}\nonumber
& \left|(\Phi, \left[ (N-1)^{-1}v_p U(a^{\dagger}_p a_pa_0^{\dagger}a_0+h.c.) U^{\dagger}-2a^{\dagger}_p a_p\right] \Phi)\right |
\\ & =\left|(\Phi,(v_p a^{\dagger}_pa_p\left(\frac{N-N_+}{N-1}-1\right)+h.c.)\Phi)\right| 
\leq v_0 \frac{(\Phi,N_+(N_+-1)\Phi)}{N-1}.
\end{align} 
By combining these inequalities, we obtain the desired bound. 
\end{proof}
The main result of this section, Proposition \ref{prop}, is a direct consequence of the last two Lemmas.

\section{Proof of Theorem 1}

For brevity we denote  $H_N-E_{\rm H}(N)$ by $H_N'$.

\subsection{Lower bound}
Let $\xi> 0$ and consider $i$ such that  $e_i(H_N)-E_{\rm H}(N)\leq \xi$. Let  $G$ be the span of the $i+1$ lowest eigenvectors of $H'_N$ (their existence is shown in Theorem~\ref{thm2}; its proof relies on compactness arguments and does not exploit Theorem \ref{thm1}). For any normalized $\Psi \in G$, $(\Psi,  H'_N \Psi)\leq e_i(H'_N)$. For $\Psi \in G$, let 
 $\Phi=U\Psi \in L^2(\mathbb{T}^d)\otimes \mathcal{F}_+^{\leq N}$. With the choice $\epsilon=\sqrt{\xi/N}$ in Proposition~\ref{prop} it  follows, by additionally invoking Lemma~\ref{lemma2}, that $(\Phi,UH_N'U^{\dagger}\Phi)\geq (\Phi,\mathbb{H}^{\text{F}}\Phi) -\frac{C \xi^{3/2}}{\sqrt{N}}$ for some $C>0$. Thus clearly $e_i(H'_N)+C \xi^{3/2} N^{-1/2}\geq \max_{\Psi \in  G } (\Psi,U^{\dagger}\mathbb{H}^{\text{F}} U\Psi) $ and, by the min-max principle, 
\begin{equation}\label{low}
e_i(H'_N)+C \xi^{3/2} N^{-1/2}\geq e_i(\mathbb{H}^{\text{F}}).
\end{equation}

\subsection{Upper bound}

For the upper bound, we use \emph{Fock space localization}. It is quantified by the following result \cite{LNSS, LiebSolovej}.

\begin{proposition}\label{lnss}
Let $A>0$ be an operator on $\mathcal{F}$ with domain $D(A)$ such that for the projections $\bar{P}_j: \mathcal{F}\rightarrow \mathcal{H}_j$ we have $\bar{P}_jD(A)\subset D(A)$ and $\bar{P}_j A \bar{P}_i=0$ for $|i-j|>\sigma$ for some constant $\sigma>0$. Then, if $f,g \in C^{\infty}(\mathbb{R}, \mathbb{R}_{\geq 0})$ with $f^2+g^2\equiv 1$ and $f(x)=1$ for $|x|\leq 1/2$ as well as $f=0$ for $x>1$, then 
 we have the inequality \begin{equation}
-\frac{C \sigma^3}{M^2}\sum_{j=0}^{\infty}\bar{P}_j A \bar{P}_j\leq  A-f_M A f_M -g_M A g_M\leq \frac{C \sigma^3}{M^2}\sum_{j=0}^{\infty}\bar{P}_j A \bar{P}_j 
\end{equation} for all $M\in \mathbb{N}$. Here $f_M$ denotes the operator
\begin{equation}
f_M:= \sum_{j=0}^{\infty}f\left(\frac{j}{M}\right)\bar{P}_j
\end{equation} 
and analogously for $g_M$. 
\end{proposition}

For the proof, which is based on an IMS-type argument, see \cite[Appendix B]{LNSS}. Proposition~\ref{lnss} can be used to quantify the error made  by constraining the states on  Fock space to contain only up to $M$ particles. From the Proposition, we deduce 

\begin{lemma}\label{localization}
We have
\begin{equation}
\mathbb{H}^{\rm{F}}-f_M \mathbb{H}^{\rm{F}} f_M -g_M \mathbb{H}^{
\rm{F}} g_M \geq - \frac{C}{M^2}(\mathbb{H}^{\rm{F}}+C)
\end{equation}
for all $M \in \mathbb{N}$. 
\end{lemma}

\begin{proof}
We apply Proposition \ref{lnss} for $A=\mathbb{H}^{\rm{F}}-e_0(\mathbb{H}^{\rm{F}})$.  From Lemma \ref{lemma3} it follows $e_0(\mathbb{H}^{\rm{F}})\geq -C_2/C_1$ and further that $\sum_j \bar{P}_j (\mathbb{H}^{\rm{F}}-e_0(\mathbb{H}^{\rm{F}})\bar{P}_j=\mathbb{H}_0^{\rm{F}}-e_0(\mathbb{H}^{\rm{F}})  \leq C_1 C_0^{-1} \mathbb{H}^{\rm{F}}+(C_2 C_0^{-1}-e_0(\mathbb{H}^{\rm{F}}))$, which leads to the right hand side of the claimed inequality, with $\sigma=2$. Using $f^2_M+g^2_M=\mathbb{I}$, we have $A-f_MAf_M-g_MAg_M=\mathbb{H}^{\rm{F}}-f_M \mathbb{H}^{\rm{F}} f_M -g_M \mathbb{H}^{\rm{F}} g_M$, which yields the left hand side of the desired result. 
\end{proof}

\begin{lemma}\label{cor}
Let $Y\subset L^2(\mathbb{T}^d)\otimes \mathcal{F}_+$ be the spectral subspace of $\mathbb{H}^{\rm{F}}$ corresponding to an energy window $ [e_0(\mathbb{H}^{\rm F}), e_0(\mathbb{H}^{\rm F})+\xi]$ for $\xi>0$. Then $\dim f_NY:=\dim \lbrace f_N\Psi: \Psi\in Y \rbrace=\dim Y$ for   $N$ large enough and  $\frac{\xi}{N}$ small enough.
\end{lemma}

\begin{proof}
Suppose $\dim f_N Y < \dim Y$, in which case there exists $\Phi\in Y$ with $\|\Phi\|=1$ such that $f_N\Phi=0$. In particular, $\Phi = g_N \Phi$. 
From Lemma~\ref{lemma3} we thus conclude that
\begin{equation}
e_0(\mathbb{H}^{\rm{F}}) + \xi \geq (\Phi, \mathbb{H}^{\rm{F}}\Phi) = (\Phi, g_N \mathbb{H}^{\rm{F}} g_N \Phi) \geq C (\Phi,g_NN_+g_N\Phi)-C \geq C N - C , 
\end{equation} 
which is a contradiction for large $N$ and small $\xi/N$.  
\end{proof}

Let us now take $Y\subset L^2(\mathbb{T}^d)\otimes \mathcal{F}_+$ to be the spectral subspace of $\mathbb{H}^{\rm{F}}$ corresponding to energies $E\leq e_i(\mathbb{H}^{\text{F}})$, and let $1\leq \xi \leq N$. The  bound \eqref{low} together with the upper bound of Lemma~\ref{lemma1}  implies that $e_i(\mathbb{H}^{\rm{F}}) \leq C \xi$, and hence also $(\Phi, (\mathbb{H}^{\rm{F}})^k \Phi)\leq  C \xi^k$ for $k=1,2$ for any $\Phi \in Y$.   By Lemma \ref{localization} and Proposition \ref{prop} (with the choice $\epsilon=\sqrt{\xi/N})$ we have
\begin{equation}
\mathbb{H}^{\text{F}}\geq f_N U H'_N U^{\dagger}f_N +e_0(\mathbb{H}^{\rm{F}})g^2_N-\frac{C}{N^2}(\mathbb{H}^{\rm{F}}+K)-C\frac{f_N N_+^2 f_N}{\sqrt{N\xi}}-C\sqrt{\frac{\xi}{N}}f_N N_+ f_N.
\end{equation} 
By taking the expectation value in any normalized $\Phi \in Y$, we obtain, by Lemmas~\ref{lemma3} and~\ref{lemma4} and the simple inequalities $N_+^k\geq f_N N_+^k f_N$ for $k=1,2$, the bound 
\begin{equation}\label{upeq}
C \xi \left(\frac{\xi}{N}\right)^{1/2} + e_i(\mathbb{H}^{\rm{F}}) \geq (\Phi, f_N U H'_N U^{\dagger} f_N \Phi)+e_0(\mathbb{H}^{\rm{F}})(\Phi,g_N^2\Phi) .
\end{equation} 
Since $g^2(x)\leq 2x$, we have $g_N^2\leq \frac{2N_+}{N}\leq \frac{C \mathbb{H}^{\rm{F}}+C}{N}$ by Lemma~\ref{lemma3}. For $Y\in \Phi$ we thus have $(\Phi, g^2_N \Phi)\leq \frac{C\xi+C}{N}$.  
Hence  $1\geq (\Psi,f^2_N \Psi) \geq 1- \frac{C\xi+C}{N} >0$ for large $N$ and $\xi/N$ small enough. 
By Lemma~\ref{cor} and the min-max principle, the maximum over $Y$ of the right hand side \eqref{upeq}  is at least as large as $e_i(H_N')+O(\xi^2N^{-1})$. This 
allows us to conclude that 
\begin{equation}\label{up}
C\xi \left(\frac{\xi}{N}\right)^{1/2} +e_i(\mathbb{H}^{\rm{F}})\geq e_i(H'_N)
\end{equation}  
for some $C>0$, which is the desired bound.

\section{Proof of Theorem 2}

\subsection{Existence of eigenvectors}

We shall now conclude the existence of eigenvectors of $H_N$ and $\mathbb{H}^{\text{F}}$ by showing that these operators have compact resolvents. By the definition of compactness and the spectral theorem one  easily sees that if $A\geq B > 0$, then the compactness of $B^{-1}$ implies the compactness of $A^{-1}$. Since  the particles are confined to the unit torus, for any $\epsilon>0$ the operators $T+\epsilon$ and $P^2+\epsilon$ are strictly positive and have purely discrete spectra with eigenvalues accumulating at infinity; therefore, they have a compact inverse. The same observation applies to the operator \begin{equation}
\mathbb{H}_0:=\frac{P^2}{2M}+\sum_{p\neq 0} e_p b^{\dagger}_p b_p
\end{equation} since $\lim_{|p|\rightarrow\infty} e_p=\infty$ and $\inf_p e_p>0$. Since $H_N\geq T +\frac{P^2}{2M}$, we conclude  that $H_N$ has compact resolvent, which, by the spectral theorem,  implies that the spectrum of $H_N$ is discrete and eigenvectors exist. On the other hand, by completing the square, as in Lemma~\ref{lemma3}, it is easy to see that 
\begin{equation}
\mathbb{H}^{\text{F}}\geq c \mathbb{H}_0 -d
\end{equation} 
for appropriate constants $c, d>0$. The existence of eigenvectors of $\mathbb{H}^{\text{F}}$, along with the fact that its spectrum is discrete, follows now from precisely the same reasoning as above. This proves the first part of Theorem 2.
 
\subsection{Convergence of eigenvectors}

Fix $\xi>0$ and take any $i$ such that $e_i(H'_N)\leq \xi$, uniformly in $N$. Recall that from the proof of the lower bound in Theorem~\ref{thm1}, we have
$\sum_{j=0}^i e_j(\mathbb{H}^{\text{F}})\leq \sum_{j=0}^i (U\Psi_j, \mathbb{H}^{\text{F}} U\Psi_j)\leq \sum_{j=0}^i e_j(H_N')+c_N$ with $\lim_N c_N=0$ for $i$ fixed. The upper bound \eqref{up} implies further that $e_j(H_N') \leq e_j(\mathbb{H}^{\text{F}}) +c'_N$ where again $c'_N$ goes to zero as $N\rightarrow \infty$. Thus, 
\begin{equation}\label{last}
\lim_{N\rightarrow \infty} \sum_{j=0}^i (U\Psi_j, \mathbb{H}^{\text{F}} U\Psi_j)= \sum_{j=0}^i e_j(\mathbb{H}^{\text{F}}).
\end{equation}
We first show the convergence for ground states. Recall that $\mathbb{P}_i$ denotes the orhogonal projection onto the eigenspace of $\mathbb{H}^{\rm{F}}$ corresponding to energy $e_i(\mathbb{H}^{\rm{F}}$). By writing $U\Psi_0=a_N+b_N$, $a_N\in\text{ran}\mathbb{P}_0$  and $b_N \perp a_N$, we have 
\begin{equation}
(U\Psi_0, \mathbb{H}^{\text{F}} U \Psi_0)  \geq \|\Psi_0\|^2 e_0(\mathbb{H}^{\text{F}})+\left(\inf_{\Psi \in \text{ker}{\mathbb{P}}_0} (\Psi,\mathbb{H}^{\text{F}}\Psi)-e_0(\mathbb{H}^{\text{F}})\right)\|b_N\|^2. 
\end{equation}
By using \eqref{last} for $i=0$ as well as the fact that $\inf_{\Psi \in \text{ker}\mathbb{P}_0 }(\Psi,\mathbb{H}^{\text{F}}\Psi)>e_0(\mathbb{H}^{\text{F}})$ by the discreteness of the spectrum of $\mathbb{H}^{\text{F}}$, we have $\lim_{N\rightarrow \infty} \|b_N\|=0$, which is the desired result for the ground states. 

For higher eigenvectors, we apply a reasoning similar to the one in \cite[Sec. 5]{SeYin}. 
Let us take any $k>0$ such that $e_{k+1}(\mathbb{H}^{\text{F}})>e_k(\mathbb{H}^{\text{F}})$. Consider the operator $\tilde{H}:=\mathbb{H}^{\text{F}} \tilde{\mathbb{P}}_k +e_{k}(\mathbb{H}^{\text{F}})(1-\tilde{\mathbb{P}}_k)$ where $\tilde{\mathbb{P}}_k$ denotes the projection onto the $k+1$ lowest eigenvectors of the Fr\"ohlich Hamiltonian $\mathbb{H}^{\rm{F}}$. $\tilde{H}$ acts on $L^2(\mathbb{T}^d)\otimes \mathcal{F}_+$ and has spectrum $\lbrace e_0(\mathbb{H}^{\text{F}}), ..., e_k(\mathbb{H}^{\text{F}})\rbrace$. Therefore, by the min-max principle, 
\begin{equation}
\sum_{i=0}^k (U\Psi_i, \tilde{H} U\Psi_i) \geq \sum_{i=0}^k e_i(\mathbb{H}^{\text{F}}). 
\end{equation} 
Clearly, $\mathbb{H}^{\text{F}}\geq \mathbb{H}^{\text{F}}\tilde{\mathbb{P}}_k +e_{k+1}(\mathbb{H}^{\text{F}})(1-\tilde{\mathbb{P}}_k)$ so that 
\begin{equation}\label{bul}
\sum_{i=0}^k (U\Psi_i, \mathbb{H}^{\text{F}} U\Psi_i) \geq \sum_{i=0}^k e_i(\mathbb{H}^{\text{F}}) +(e_{k+1}(\mathbb{H}^{\text{F}})-e_k(\mathbb{H}^{\text{F}}))\sum_{i=0}^k \|(1-\tilde{\mathbb{P}}_k) U\Psi_i \|^2,
\end{equation}  
which can be rewritten as 
\begin{equation}\label{rewr}
\sum_{i=0}^k (U\Psi_i, \tilde{\mathbb{P}}_k U \Psi_i) \geq k+1-\frac{\sum_{i=0}^k\left(e_i(\mathbb{H}^{\text{F}})-(\Psi_i, U^{\dagger}\mathbb{H}^{\text{F}} U\Psi_i)\right)}{e_{k+1}(\mathbb{H}^{\text{F}})-e_k(\mathbb{H}^{\text{F}})}. 
\end{equation} 
Note that the last term converges to zero as $N\to \infty$ by \eqref{last}. 
Take now $l$ to be the largest integer such that  $e_l(\mathbb{H}^{\text{F}})<e_k(\mathbb{H}^{\text{F}})$. The dimension of the eigenspace corresponding to $e_k(\mathbb{H}^{\text{F}})$  therefore equals $k-l$. We have the simple identity \begin{equation}\label{simp}
\sum_{i=l+1}^k (U\Psi_i, \mathbb{P}_k U \Psi_i) = \sum_{i=0}^k (U\Psi_i, \tilde{\mathbb{P}}_k U\Psi_i) + \sum_{i=0}^l (U\Psi_i,\tilde{\mathbb{P}}_l U\Psi_i) -\sum_{i=0}^k (U\Psi_i,\tilde{\mathbb{P}}_l U\Psi_i)-\sum_{i=0}^l (U\Psi_i, \tilde{\mathbb{P}}_k U\Psi_i)
\end{equation} (note the presence of both {tilded} and {untilded} operators). For the first two terms, we can use \eqref{rewr} for a lower bound. Moreover, since the $\Psi_i$ are orthonormal, we have $\sum_{i=0}^k(U\Psi_i, \tilde{\mathbb{P}}_l U\Psi_i) \leq {\rm Tr}\, \tilde{\mathbb{P}}_l = l+1$. The last term in \eqref{simp} is trivially bounded from below by  $-(l+1)$. We thus conclude that  
\begin{equation}
k-l\geq \sum_{i=l+1}^k (U\Psi_i, \mathbb{P}_k U \Psi_i) \geq k-l -C_N-D_N,
\end{equation} 
where the quantities $C_N>0,D_N>0$ can be read off from \eqref{rewr}  and vanish as $N\rightarrow\infty$, because of \eqref{last}. Therefore, $\sum_{i=l+1}^k (U\Psi_i, \mathbb{P}_k U \Psi_i)\rightarrow k-l$, but as each individual term in the sum is $\leq 1$, we must have $\lim (U\Psi_i \mathbb{P}_k U\Psi_i)=1$ for every eigenstate of $H_N'$ with energy $e_k(H'_N)$. This is precisely the convergence result stated in Theorem 2, whose proof is now complete. 

\bigskip
\noindent \textbf{Acknowledgments.} Financial support through the European Research Council (ERC) under the European Union's Horizon 2020 research and innovation programme grant agreement No 694227 (R.S.) and  the Maria Sk\l odowska-Curie grant agreement No. 665386 (K.M.)  is gratefully acknowledged.

\end{document}